\pdfoutput=1
\documentclass[11pt]{amsart}
\usepackage{xcolor}
\usepackage[margin=1in,includefoot]{geometry}
\usepackage[utf8]{inputenc}
\usepackage[T1]{fontenc}
\usepackage{lmodern,microtype}
\usepackage{mathtools,amssymb,amsthm}
\usepackage{tikz}
\usepackage[unicode]{hyperref}
\usepackage{bookmark}

\makeatletter
\def\@makefntext{\leftskip 7pt\parindent 0pt\leavevmode\llap{\@makefnmark\kern 3pt}}
\let\@adminfootnotes\relax
\makeatother

\newtheorem{theorem}{Theorem}
\newtheorem{lemma}[theorem]{Lemma}
\newtheorem{corollary}[theorem]{Corollary}
\newtheorem{claim}[theorem]{Claim}
\newtheorem{proposition}[theorem]{Proposition}

\def\famC{\mathcal{C}}
\def\famG{\mathcal{G}}
\def\famP{\mathcal{P}}
\def\famR{\mathcal{R}}
\def\famS{\mathcal{S}}
\def\famT{\mathcal{T}}
\def\famV{\mathcal{V}}

\def\famX{\mathcal{X}}

\def\MaxCliques{\mathfrak{M}}

\def\setN{\mathbb{N}}
\def\setR{\mathbb{R}}

\let\leq\leqslant
\let\geq\geqslant
\let\setminus\smallsetminus
\let\epsilon\varepsilon

\allowdisplaybreaks

\begin{document}

\pagestyle{firstpage}

\renewcommand*{\thefootnote}{\alph{footnote}}

\title{Coloring and Maximum Weight Independent Set of~Rectangles}
\author[Parinya Chalermsook]{Parinya Chalermsook\protect\footnotemark}
\footnotetext[1]{Aalto University, Finland (\texttt{parinya.chalermsook@aalto.fi}). Supported by European Research Council (ERC) under the European Union's Horizon 2020 research and innovation programme (grant agreement No.\ 759557) and by Academy of Finland Research Fellowship, under grant number 310415.}

\author[Bartosz Walczak]{Bartosz Walczak\protect\footnotemark}
\footnotetext[2]{Jagiellonian University, Kraków, Poland (\texttt{walczak@tcs.uj.edu.pl}). Partially supported by National Science Center of Poland grant 2015/17/D/ST1/00585.}

\begin{abstract}
In 1960, Asplund and Grünbaum proved that every intersection graph of axis-parallel rectangles in the plane admits an $O(\omega^2)$-coloring, where $\omega$ is the maximum size of a clique.
We present the first asymptotic improvement over this six-decade-old bound, proving that every such graph is $O(\omega\log\omega)$-colorable and presenting a polynomial-time algorithm that finds such a coloring.
This improvement leads to a polynomial-time $O(\log\log n)$-approximation algorithm for the maximum weight independent set problem in axis-parallel rectangles, which improves on the previous approximation ratio of $O(\frac{\log n}{\log\log n})$.

\end{abstract}

\maketitle

\renewcommand*{\thefootnote}{\arabic{footnote}}
\setcounter{footnote}{0}

\section{Introduction}

\subsection*{Coloring of Rectangles}

Let $\famR$ be a family of axis-parallel rectangles in the plane.
The chromatic number of $\famR$, denoted by $\chi(\famR)$, is the minimum number of colors that can be assigned to the rectangles so that any two intersecting rectangles receive different colors.
The clique number of $\famR$, denoted by $\omega(\famR)$, is the maximum size of a set $\famC\subseteq\famR$ such that any two rectangles in $\famC$ intersect.
These two terms are equivalent to the chromatic number $\chi(G)$ and the clique number $\omega(G)$ of the intersection graph $G$ of $\famR$.
Since $\chi(G)\geq\omega(G)$, a natural question is whether $\chi(G)$ can be bounded from above in terms of $\omega(G)$.
This question in various graph classes has received a lot of attention from discrete mathematics community, and it has also played crucial roles in the theory of algorithms and mathematical programming.

In general, it is well known that triangle-free graphs (that is, graphs with clique number $2$) can have arbitrarily large chromatic number~\cite{erdos1959graph}.
Classes of graphs $\famG$ that admit a function bounding $\chi(G)$ in terms of $\omega(G)$ for every $G\in\famG$ are called $\chi$-bounded.
There has been immense progress in the study of $\chi$-bounded classes of graphs in recent years---see the survey by Scott and Seymour~\cite{scott?survey} and the references therein.
In particular, various classes of geometric intersection graphs are known to be $\chi$-bounded or not $\chi$-bounded---see e.g.\ \cite{kostochka2004coloring,pawlik2014triangle,rok2019outerstring}.
The history of this question for rectangle intersection graphs dates back to 1948, when Bielecki~\cite{bielecki1948problem} asked whether triangle-free rectangle graphs have bounded chromatic number.
Asplund and Grünbaum~\cite{asplund1960coloring} not only answered this question in the positive but also showed a more general bound of $\chi(\famR)\leq 4\omega(\famR)^2-4\omega(\famR)$.
This bound was later improved by Hendler~\cite{hendler1998schranken} to $\chi(\famR)\leq 3\omega(\famR)^2-2\omega(\famR)-1$.
Kostochka~\cite{kostochka2004coloring} constructed rectangle families $\famR$ with $\chi(\famR)=3\omega(\famR)$, and this remains the best known lower bound.
Chalermsook~\cite{chalermsook2011coloring} proved the bound $\chi(\famR)=O(\omega(\famR)\log\omega(\famR))$ for the special case that $\famR$ contains no nested pair of rectangles.
Closing or even narrowing down the gap between the linear lower bound and the quadratic upper bound for the general families of rectangles has been a long-standing open problem.

\subsection*{Maximum Weight Independent Set of Rectangles (MWISR)}
In MWISR, we are given a family of $n$ axis-parallel rectangles in the plane together with weights assigned to them, and we aim at finding a maximum weight subfamily (called an independent set or a packing) that contains no two intersecting rectangles.
Besides being a fundamental problem in geometric optimization, MWISR is interesting from several perspectives.
First, it arises in various applications, including map labeling~\cite{agarwal1998label,doerschler1992rule}, resource allocation~\cite{lewin2002routing}, data mining~\cite{lent1997clustering,khanna1998approximating,fukuda1996data}, and unsplittable flow routing~\cite{bonsma2014constant}.
Second, it is one of the ``somewhat tractable'' special cases of the general maximum weight independent set problem: given an $n$-vertex graph with weights on the vertices, find a maximum weight subset of the vertices containing no two vertices connected by an edge.
This problem for general graphs is NP-hard to approximate to within a factor of $n^{1-\varepsilon}$ for every $\varepsilon>0$~\cite{hastad1996clique,zuckerman2007linear}, with the best known approximation factor being $O(n(\log\log n)^2/\log^3n)$~\cite{feige2004approximating}.
In special graph classes defined by intersections of geometric objects (such as disks, squares, and more generally---fat objects), polynomial-time approximation schemes (PTASes) are known~\cite{chan2003polynomial,erlebach2005polynomial}.
Rectangles are perhaps the simplest natural objects for which the maximum independent set problem is not known to admit a PTAS\@.
MWISR is NP-hard~\cite{fowler1981optimal} and there have been active attempts in the past decade from various groups of researchers on obtaining approximation algorithms.
The best known approximation factor is $O(\frac{\log n}{\log\log n})$ by Chan and Har-Peled~\cite{chan2012approximation}.
In the unweighted case, Chalermsook and Chuzhoy~\cite{chalermsook2009maximum} presented an $O(\log\log n)$-approximation.
Recently, a quasi-polynomial-time approximation scheme (QPTAS) was presented by Adamaszek, Har-Peled, and Wiese~\cite{adamaszek2019approximation,adamaszek2013approximation} (see also an improvement on the unweighted case by Chuzhoy and Ene~\cite{chuzhoy2016approximating}).
Obtaining a PTAS or even a polynomial-time constant-factor approximation for MWISR remains an elusive open problem.

\subsection*{Connections between Coloring and MWISR}

These two problems are related through the perspective of mathematical programming.
In particular, consider the clique-constrained independent set polytope of a graph $G$:
\[\mathsf{QSTAB}(G) = \Bigl\{x\in\setR^{V(G)}\colon x\geq 0\enspace\text{and}\enspace\sum_{v\in Q}x_v\leq 1\enspace\text{for every clique $Q$ in $G$}\}.\]
For a graph $G$ and a weight vector $w\in\setR^{V(G)}$, let $\mathsf{FRAC}(G,w)=\max\{w\cdot x\colon x\in\mathsf{QSTAB}(G)\}$ and $\mathsf{INT}(G,w)=\max\{w\cdot x\colon x\in\mathsf{QSTAB}(G)\cap\{0,1\}^{V(G)}\}$, the latter being the maximum weight of an independent set in $G$ with respect to the weights $w$.
Clearly, $\mathsf{INT}(G,w)\leq\mathsf{FRAC}(G,w)$.
The integrality ratio (or integrality gap) $\mathsf{gap}(G,w)$ is the ratio $\frac{\mathsf{FRAC}(G,w)}{\mathsf{INT}(G,w)}$.
Since a fractional solution $x \in \mathsf{QSTAB}(G)$ with value $w\cdot x\geq\mathsf{INT}(G,w)$ can be found efficiently\footnote{In general graphs, it can be computed via SDP, as a solution optimizing $w\cdot x$ over the Lovász theta body of $G$.}, rounding this LP solution is a natural algorithmic paradigm for approximating the maximum weight independent set problem, especially in restricted graph classes.

The integrality ratio of $\mathsf{QSTAB}$ has a strong connection to certain Ramsey-type bounds.
More formally, let $\mathcal{G}$ be any graph class that is closed under clique replacement operation\footnote{This holds for various natural graph classes such as perfect graphs and geometric intersection graphs.}.
When $w = 1$ (the unweighted case), proving the upper bound $\mathsf{gap}(G,w) \leq \gamma$ for all $G \in \mathcal{G}$, is equivalent to proving the upper bound $R(s,t) \leq \gamma s(t-1)$ on the Ramsey numbers\footnote{The Ramsey number $R(s,t)$ is the minimum integer $n$ such that every $n$-vertex graph contains a clique of size $s$ or an independent set of size $t$.} for all graphs in the same graph class $\mathcal{G}$.
When allowing an arbitrary weight function $w$, proving  $\mathsf{gap}(G,w) \leq \gamma$ for all $G \in \mathcal{G}$, is equivalent to upper bounding the ratio $\frac{\chi_f(G)}{\omega(G)} \leq \gamma$ for all $G \in \mathcal{G}$.
These connections are constructive~\cite{chalermsook2016note}.
Therefore, one way to design an efficient approximation algorithm for the maximum independent set problem in any graph class $\mathcal{G}$ is to prove an (algorithmic) upper bound on $\frac{\chi(G)}{\omega(G)}$ for graphs in the same graph class.

The polytope $\mathsf{QSTAB}(G)$ has played crucial roles from both algorithms and mathematical optimization perspectives; a notable example is its application to finding maximum cliques and independent sets in perfect graphs.
It is particularly appealing for rectangle intersection graphs $G$, which have only $O(n^2)$ maximal cliques.
For these graphs, an LP over $\mathsf{QSTAB}(G)$ can be explicitly written and solved by a near-linear-time algorithm~\cite{chekuri2020fast}.
Therefore, it is an interesting question on its own to pinpoint the value of $\mathsf{gap}(G,w)$ for rectangle graphs.

\subsection*{Our Contributions}

First, we present the following improvement on the $O(\omega^2)$ coloring bound of Asplund and Grünbaum~\cite{asplund1960coloring}.

\begin{theorem}
\label{thm:coloring}
Every family of axis-parallel rectangles in the plane with clique number\/ $\omega$ is\/ $O(\omega\log\omega)$-colorable, and an\/ $O(\omega\log\omega)$-coloring of it can be computed in polynomial time.
\end{theorem}

Second, via a simple reduction, we obtain the following result for MWISR\@.
We remark that the reduction was used implicitly in the paper of Chalermsook and Chuzhoy~\cite{chalermsook2009maximum}.

\begin{theorem}
\label{thm:ind set}
There is a polynomial-time\/ $O(\log\log n)$-approximation algorithm for MWISR, and the integrality ratio of the clique-constrained LP for rectangle graphs is at most\/ $O(\log\log n)$.
\end{theorem}

This result improves upon the $O(\frac{\log n}{\log\log n})$-approximation by Chan and Har-Peled~\cite{chan2012approximation} for MWISR.
It also substantially simplifies and derandomizes the known $O(\log\log n)$-approximation in the unweighted setting~\cite{chalermsook2009maximum}.
The bound on the integrality ratio combined with a fast LP solver from~\cite{chekuri2020fast} imply that an $O(\log\log n)$ estimate on the value of MWISR can be computed in $O(n^2\operatorname{polylog}n)$ time.

The main new technical ingredient of this paper is a ``hierarchical decomposition'' of a family of rectangles, inspired by the work of Kierstead and Trotter~\cite{kierstead1981extremal}.
In section~\ref{sec:warm-up}, we present a small ``warm-up'' result that highlights the main idea behind this decomposition.
In section~\ref{sec:coloring}, we define the decomposition and use it to prove Theorem~\ref{thm:coloring}.
We present the proof of Theorem~\ref{thm:ind set} in section~\ref{sec:ind set}.

\section{Preliminaries} 

\subsection*{Definitions}

A \emph{rectangle} is a closed set of the form $[a,b]\times[c,d]$ in the plane, where $a<b$ and $c<d$.
The \emph{width} and the \emph{height} of such a rectangle are the values $b-a$ and $d-c$, respectively.

Let $\famR$ be a family of rectangles.
The \emph{intersection graph} of $\famR$ has vertex set $\famR$ and edge set defined as follows: two rectangles $R,R'\in\famR$ are connected by an edge if they intersect, that is, if $R\cap R'\neq\emptyset$.
A subfamily $\famC$ of $\famR$ is a \emph{clique} if the intersection of all rectangles in $\famC$ is non-empty.
This is the same as to say that every pair of rectangles in $\famC$ intersects, so this notion of a clique corresponds to a clique in the intersection graph of $\famR$.
We say that a clique $\famC$ in $\famR$ contains a point $p$ if $p$ belongs to every rectangle in $\famC$.
We let $\omega(\famR)$ denote the \emph{clique number} of $\famR$, that is, the maximum size of a clique in $\famR$.
A subfamily $\famS$ of $\famR$ is an \emph{independent set} if the rectangles in $\famS$ are pairwise disjoint.
A \emph{coloring} of $\famR$ is an assignment of colors to the rectangles in $\famR$ such that the rectangles of any given color form an independent set.
These notions correspond to independent sets and colorings of the intersection graph of $\famR$.
We say that $\famR$ is \emph{$k$-colorable} if there is a coloring of $\famR$ using $k$ colors.
We let $\chi(\famR)$ denote the \emph{chromatic number} of $\famR$, that is, the minimum number $k$ such that $\famR$ is $k$-colorable.

Let $R$ and $R'$ be two rectangles in the plane that intersect ($R\cap R'\neq\emptyset$).
We distinguish several possible types of intersections; see Figure~\ref{fig:intersection}.
If $R$ contains at least one corner of $R'$ or vice versa, then we have a \emph{corner intersection} between $R$ and $R'$.
Otherwise, we have a \emph{crossing intersection} between $R$ and $R'$, and we say that $R$ and $R'$ \emph{cross}.
We have a \emph{containment intersection} when one rectangle contains the other (which is a particular case of a corner intersection).
We have a \emph{vertical intersection} if one rectangle intersects both the top and the bottom sides of the other.

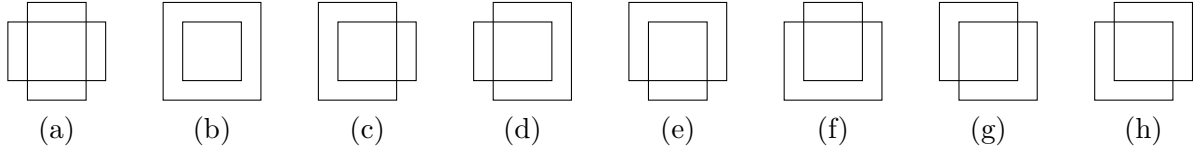
\begin{figure}[t]
    \centering
    \def\interskip{\hskip 0.75cm\relax}
\tikzstyle{every picture}=[scale=1.3]
\tikzstyle{every node}=[inner sep=6pt]
\begin{tikzpicture}
  \draw (0,0.8) rectangle (1,0.2);
  \draw (0.2,1) rectangle (0.8,0);
  \node[below] at (0.5,0) {(a)};
\end{tikzpicture}\interskip
\begin{tikzpicture}
  \draw (0,1) rectangle (1,0);
  \draw (0.2,0.8) rectangle (0.8,0.2);
  \node[below] at (0.5,0) {(b)};
\end{tikzpicture}\interskip
\begin{tikzpicture}
  \draw (0,1) rectangle (0.8,0);
  \draw (0.2,0.8) rectangle (1,0.2);
  \node[below] at (0.5,0) {(c)};
\end{tikzpicture}\interskip
\begin{tikzpicture}
  \draw (0,0.8) rectangle (0.8,0.2);
  \draw (0.2,1) rectangle (1,0);
  \node[below] at (0.5,0) {(d)};
\end{tikzpicture}\interskip
\begin{tikzpicture}
  \draw (0,1) rectangle (1,0.2);
  \draw (0.2,0.8) rectangle (0.8,0);
  \node[below] at (0.5,0) {(e)};
\end{tikzpicture}\interskip
\begin{tikzpicture}
  \draw (0,0.8) rectangle (1,0);
  \draw (0.2,1) rectangle (0.8,0.2);
  \node[below] at (0.5,0) {(f)};
\end{tikzpicture}\interskip
\begin{tikzpicture}
  \draw (0,1) rectangle (0.8,0.2);
  \draw (0.2,0.8) rectangle (1,0);
  \node[below] at (0.5,0) {(g)};
\end{tikzpicture}\interskip
\begin{tikzpicture}
  \draw (0,0.8) rectangle (0.8,0);
  \draw (0.2,1) rectangle (1,0.2);
  \node[below] at (0.5,0) {(h)};
\end{tikzpicture}
\vspace*{-0.5ex}
    \caption{All possible ways a pair of rectangles can intersect: (a) a crossing intersection, (b)--(h) corner intersections (each involving at least two corners), (b) a containment intersection, (a)--(d) vertical intersections.}
    \label{fig:intersection}
\end{figure}

Let us fix a family $\famR$ of $n$ rectangles that is the input to our problem.
For each rectangle $R\in\famR$, let $\famV(R)$ denote the rectangles in $\famR\setminus\{R\}$ that intersect both the bottom and the top sides of $R$, and let $\famX(R)$ denote the rectangles in $\famV(R)$ that cross $R$.
Thus, if $R,R'\in\famR$ and the height of $R'$ is greater than the height of $R$, then the following holds:
\begin{itemize}
\item there is a vertical intersection between $R$ and $R'$ if and only if $R'\in\famV(R)$;
\item there is a crossing intersection between $R$ and $R'$ if and only if $R'\in\famX(R)$.
\end{itemize}
It is important to observe that if $R'\in\famX(R)$, then $\famV(R')\subseteq\famV(R)$.

\subsection*{Preliminary Results}

A family of rectangles $\famR$ is \emph{$s$-sparse} if one can fix $s$ points $p^R_1,\ldots,p^R_s$ in each rectangle $R\in\famR$ so that the intersection $R\cap R'$ of any two crossing rectangles $R,R'\in\famR$ contains at least one of the points $p^R_1,\ldots,p^R_s,p^{R'}_1,\ldots,p^{R'}_s$.
For instance, a family of squares is $0$-sparse, because no pair of squares can cross.
The following lemma is slightly modified from~\cite{lewin2002routing,chalermsook2011coloring}.

\begin{lemma}
\label{lem:sparse}
For each\/ $s\in\setN$, every\/ $s$-sparse family of rectangles with clique number\/ $\omega$ is\/ $(2s+4)(\omega-1)$-colorable, and such a coloring can be computed in polynomial time.
\end{lemma}

\begin{proof}
Let $\famR$ be an $s$-sparse family of rectangles with clique number $\omega$.
It suffices to show that the number of edges in the intersection graph of $\famR$ is strictly less than $(s+2)(\omega-1){|\famR|}$, because this implies that there is a vertex of degree less than $(2s+4)(\omega-1)$ in every induced subgraph, which leads to a $(2s+4)(\omega-1)$-coloring by straightforward induction.

For each edge $RR'$ in the intersection graph, if $R$ and $R'$ cross, we give one token to one of the points $p^R_1,\ldots,p^R_s,p^{R'}_1,\ldots,p^{R'}_s$ that lies in $R\cap R'$.
Otherwise, $RR'$ corresponds to a corner intersection, which involves at least two corners (four in case of containment and two otherwise; see Figure~\ref{fig:intersection}), and we give half of a token to any two corners involved in the intersection.
Clearly, the total number of tokens handed out is equal to the number of edges in the intersection graph.
Moreover, for each $R$, each point $p^R_i$ receives at most $\omega-1$ tokens, and each corner of a rectangle receives at most $\frac{\omega-1}{2}$ tokens (with some corners receiving strictly fewer tokens).
Therefore, the number of edges is less than $(s+2)(\omega-1){|\famR|}$.
\end{proof}

\begin{corollary}
\label{cor:corner coloring}
Every family of rectangles with no crossing intersections and with clique number\/ $\omega$ is\/ $4(\omega-1)$-colorable, and such a coloring can be computed in polynomial time.
\end{corollary}

The following result of Asplund and Grünbaum~\cite{asplund1960coloring} will be used as a subroutine.

\begin{lemma}
\label{lem:omega square}
Every family of rectangles with clique number\/ $\omega$ is\/ $4\omega(\omega-1)$-colorable, and such coloring can be computed in polynomial time.
\end{lemma}

\begin{proof}
Let $\famR$ be a family of rectangles with clique number $\omega$.
Let $<$ be the strict partial order on $\famR$ defined so that $R<R'$ if and only if $R'\in\famX(R)$.
Every chain in the poset $(\famR,{<})$ is a clique in $\famR$, so the height of $(\famR,{<})$ is at most $\omega$.
Therefore, $\famR$ can be partitioned into $\omega$ antichains in the poset $(\famR,{<})$, that is, families $\famR_1,\ldots,\famR_\omega$ with no crossing intersections.
By Corollary~\ref{cor:corner coloring}, each of the families $\famR_1,\ldots,\famR_\omega$ is $4(\omega-1)$-colorable, so the entire family is $4\omega(\omega-1)$-colorable.
\end{proof}

\section{Warm-Up}
\label{sec:warm-up}

In this section, we present two simple coloring results that capture some of the key ideas behind our general $O(\omega\log\omega)$ bound.

\begin{proposition}
\label{prop:warmup}
Let\/ $\famR$ be a family of rectangles.
\begin{enumerate}
\item If there are only crossing and containment intersections within\/ $\famR$, then\/ $\chi(\famR)=\omega(\famR)$.
\item If there are only vertical intersections within\/ $\famR$, then\/ $\chi(\famR)\leq 3\omega(\famR)-2$.
\end{enumerate}
\end{proposition}

\begin{proof}
Let $\famR$ be a family of rectangles with only vertical intersections, and let $\omega=\omega(\famR)$.
We process the rectangles in $\famR$ in the decreasing order of their heights and put each rectangle into one of the sets $\famS_1,\ldots,\famS_\omega$ as follows.
We put a rectangle $R\in\famR$ into $\famS_i$ where $i$ is the maximum integer such that there is an \emph{$i$-witnessing clique} for $R$, that is, a clique in $\famV(R)$ containing at least one rectangle from each of the sets $\famS_1,\ldots,\famS_{i-1}$.
This is well defined---when a rectangle $R$ is being processed, the rectangles in $\famV(R)$ have been already processed and distributed to $\famS_1,\ldots,\famS_\omega$, and $i$ is always at most $\omega$ because the aforesaid clique in $\famV(R)$ together with $R$ forms a clique in $\famR$ of size at most $\omega$.
Let $\famC(R)$ denote any $i$-witnessing clique for a rectangle $R\in\famR$ where $i$ is such that $R\in\famS_i$.

First, we show that, for each $i$, there can be no crossing or containment intersection between rectangles from a single set $\famS_i$.
Suppose for the sake of contradiction that two rectangles $R,R'\in\famS_i$ form a crossing or containment intersection, where $R'$ is taller than $R$, that is, $R'\in\famV(R)$.
If $R'$ contains $R$, then $\famC(R)\cup\{R'\}$ is an $(i+1)$-witnessing clique for $R$, contradicting the assumption that $R\in\famS_i$.
If $R$ and $R'$ cross, then $\famV(R')\cup\{R'\}\subseteq\famV(R)$, and $\famC(R')\cup\{R'\}$ is an $(i+1)$-witnessing clique for $R$, again contradicting the assumption that $R\in\famS_i$.
See Figure~\ref{fig:warm-up1}.

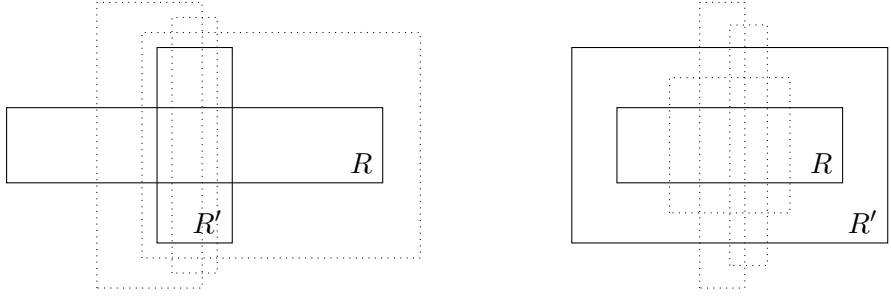
\begin{figure}[t]
    \centering
    \begin{tikzpicture}
  \draw (-2,0) rectangle (3,1);
  \draw (0,-0.8) rectangle (1,1.8);
  \draw[dotted] (-0.2,-1) rectangle (3.5,2);
  \draw[dotted] (0.2,-1.2) rectangle (0.8,2.2);
  \draw[dotted] (-0.8,-1.4) rectangle (0.6,2.4);
  \node[above left] at (3,0) {$R$};
  \node[above left] at (1,-0.8) {$R'$};
\end{tikzpicture}\hskip 2cm
\begin{tikzpicture}
  \draw (-0.6,-0.8) rectangle (3.6,1.8);
  \draw (0,0) rectangle (3,1);
  \draw[dotted] (0.7,-0.4) rectangle (2.3,1.4);
  \draw[dotted] (1.5,-1.1) rectangle (2,2.1);
  \draw[dotted] (1.1,-1.4) rectangle (1.7,2.4);
  \node[above left] at (3,0) {$R$};
  \node[above left] at (3.6,-0.8) {$R'$};
\end{tikzpicture}
    \caption{An illustration of the proof that there is no crossing or containment intersection in each set $\famS_i$.
    On the left, the dotted rectangles are the rectangles in $\famC(R')$; on the right, the dotted rectangles are those in $\famC(R)$.}
    \label{fig:warm-up1}
\end{figure}

If there are only containment and crossing intersections between the rectangles in $\famR$, then $\famS_1,\ldots,\famS_\omega$ are independent sets.
This completes the proof of statement~(1).

To prove statement~(2), we argue below that $\famS_1$ is still an independent set and that $\chi(\famS_i)\leq 3$ for $i=2,\ldots,\omega$, which together imply that $\chi(\famR)\leq 3\omega-2$.
%If there are other vertical intersections between the rectangles in $\famR$, it is still possible that such an intersection occurs between two rectangles $R,R'\in\famS_i$, where $R'\in\famV(R)$.
For $\famS_1$, if two rectangles $R,R'\in\famS_1$ intersect, where $R'\in\famV(R)$, then $\{R'\}$ is a $1$-witnessing clique for $R$, which contradicts the assumption that $R\in\famS_1$.
Therefore, $\famS_1$ is an independent set.

Now, fix $i\in\{2,\ldots,\omega\}$.
We argue that $\chi(\famS_i)\leq 3$.
Let $R\in\famS_i$.
We claim that at most one other rectangle in $\famS_i$ contains the right side of $R$.
Suppose for the sake of contradiction that two rectangles $R',R''\in\famS_i$ contain the right side of $R$, where the right side of $R'$ lies between the right sides of $R$ and $R''$.
Let $Q$ be the intersection of the rectangles in $\famC(R')$.
If $Q\cap R\neq\emptyset$, then $\famC(R')\cup\{R'\}\subseteq\famV(R)$; this implies that $\famC(R')\cup\{R'\}$ is an $(i+1)$-witnessing clique for $R$, contradicting the assumption that $R\in\famS_i$.
Thus assume $Q\cap R=\emptyset$.
This implies that $Q$ lies between the right sides of $R$ and $R''$, so $Q\cap R''\neq\emptyset$.
Now, there are two cases.
If $R''\in\famV(R')$, then $\famC(R')\cup\{R''\}$ is an $(i+1)$-witnessing clique for $R'$, contradicting the assumption that $R'\in\famS_i$.
If $R'\in\famV(R'')$, then $\famC(R')\cup\{R'\}$ is an $(i+1)$-witnessing clique for $R''$, contradicting the assumption that $R''\in\famS_i$.
See Figure~\ref{fig:warm-up2} for an illustration.
We have thus proved that at most one rectangle in $\famS_i$ other than $R$ contains the right side of $R$.
By symmetry, at most one rectangle in $\famS_i$ other than $R$ contains the left side of $R$.

\begin{figure}[t]
    \centering
    \begin{tikzpicture}
  \fill[lightgray] (1.5,-0.6) rectangle (1.8,1.4);
  \draw (0,0) rectangle (2.5,1);
  \draw (1,-0.6) rectangle (3,1.4);
  \draw (2.1,-1.2) rectangle (3.4,1.8);
  \node[above right] at (0,0) {$R$};
  \node[above left] at (3,-0.6) {$R'$};
  \node[above left] at (3.4,-1.2) {$R''$};
\end{tikzpicture}\hskip 1.5cm
\begin{tikzpicture}
  \fill[lightgray] (2.3,-0.6) rectangle (2.6,1.4);
  \draw (0,0) rectangle (2,1);
  \draw (1,-0.6) rectangle (3,1.4);
  \draw (1.7,-1.2) rectangle (3.4,1.8);
  \node[above right] at (0,0) {$R$};
  \node[above right] at (1,-0.6) {$R'$};
  \node[above left] at (3.4,-1.2) {$R''$};
\end{tikzpicture}\hskip 1.5cm
\begin{tikzpicture}
  \fill[lightgray] (2.4,-1.2) rectangle (2.7,1.8);
  \draw (0,0) rectangle (2,1);
  \draw (1,-1.2) rectangle (3,1.8);
  \draw (1.7,-0.8) rectangle (3.4,1.55);
  \node[above right] at (0,0) {$R$};
  \node[above right] at (1,-1.2) {$R'$};
  \node[above right] at (1.7,-0.8) {$R''$};
\end{tikzpicture}
    \caption{An illustration for the three cases in the proof that at most one rectangle in $\famS_i$ contains the right side of $R$.
    The gray area is the set $Q=\bigcap\famC(R')$.}
    \label{fig:warm-up2}
\end{figure}
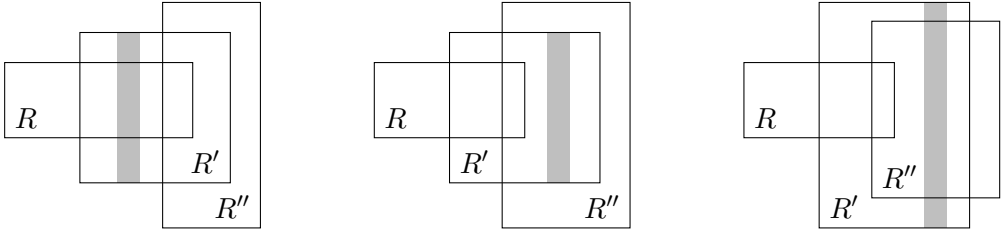

Now, we present a $3$-coloring algorithm for each set $\famS_i$.
We process the rectangles in $\famS_i$ in the decreasing order of their heights and color them greedily.
When a rectangle $R\in\famS_i$ is being processed, at most two other rectangles in $\famS_i$ have been assigned colors (one intersecting the left side and one intersecting the right side of $R$), so the greedy coloring uses at most three colors.
\end{proof}

Proposition~\ref{prop:warmup}~(1) is a strengthening of the observation by Asplund and Grünbaum \cite{asplund1960coloring} (used in the proof of Lemma~\ref{lem:omega square}) that families $\famR$ with only crossing intersections satisfy $\chi(\famR)=\omega(\famR)$.

Proposition~\ref{prop:warmup}~(2) is closely related to the result of Kierstead and Trotter \cite{kierstead1981extremal} that families of intervals on the real line with clique number $\omega$ can be colored \emph{on-line} using at most $3\omega-2$ colors.
Specifically, by a correspondence described in \cite{krawczyk2017online}, for every deterministic strategy of the adversary in the on-line coloring problem for intervals, there is an `equivalent' family of axis-parallel rectangles (with only vertical intersections) whose chromatic number is equal to the number of colors forced by that strategy against any on-line coloring algorithm.
The decomposition of the family $\famR$ into sets $\famS_1,\ldots,\famS_\omega$ used in the proof above is an adaptation of a decomposition used by Kierstead and Trotter in their proof of the upper bound of $3\omega-2$ for the on-line problem.
Kierstead and Trotter \cite{kierstead1981extremal} also showed a deterministic strategy of the adversary forcing any on-line coloring algorithm to use at least $3\omega-2$ using on a family of intervals with clique number $\omega$.
That strategy gives rise to a family of axis-parallel rectangles $\famR$ with only vertical intersections and with $\chi(\famR)=3\omega(\famR)-2$ \cite[Proposition 3.1]{krawczyk2017online}, which shows that the bound in Proposition~\ref{prop:warmup}~(2) is sharp.

\section{An \texorpdfstring{$O(\omega\log\omega)$}{O(ω log ω)}-Coloring Algorithm}
\label{sec:coloring}

Let $\famR$ be a family of rectangles and let $k=\lceil\log_2\omega(\famR)\rceil$.
Thus $\omega(\famR)\leq 2^k$.
We show how to construct a coloring of $\famR$ using $O(2^k\cdot k)$ colors (in polynomial time), which yields Theorem~\ref{thm:coloring}.

The argument consists of two steps.
In the first step, we construct a ``hierarchical decomposition'' of $\famR$ similar to the decomposition into sets $\famS_1,\ldots,\famS_\omega$ used in the proof of Proposition~\ref{prop:warmup}, but defined with a ``divide and conquer'' approach rather a simple linear induction.
This modification is essential to make it work with the second step---a ``clique reduction'' argument, which is an adaptation of an argument used before in \cite{chalermsook2011coloring,chalermsook2009maximum}.

\subsection*{Step 1: Hierarchical Decomposition}

For $i\in\setN$, let $B_i$ denote the set of binary words of length $i$, and let $\epsilon$ denote the empty binary word, so that $B_0=\{\epsilon\}$ and $B_i=\{0,1\}^i$ for $i\geq 1$.
For each $i=0,\ldots,k$, by induction, we construct a partition $\{\famS_i(w)\colon w\in B_i\}$ of $\famR$ and a non-empty set $\famP_i(R)$ of \emph{witness points} for every rectangle $R\in\famR$.

For $i=0$, let $\famS_0(\epsilon)=\famR$, and for every rectangle $R\in\famR$, let $\famP_0(R)$ be the set of intersection points of the top side of $R$ with the left and right sides of the rectangles in $\famV(R)\cup\{R\}$ (which includes the two top corners of $R$).
See Figure~\ref{fig:witness} for an illustration.

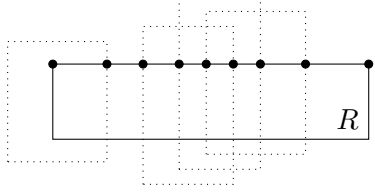
\begin{figure}[t]
    \centering
    \begin{tikzpicture}[xscale=1.2]
  \draw (0,0) rectangle (3.5,1);
  \draw[dotted] (-0.5,-0.3) rectangle (0.6,1.3);
  \draw[dotted] (1,-0.6) rectangle (2,1.5);
  \draw[dotted] (1.4,-0.4) rectangle (2.3,1.9);
  \draw[dotted] (1.7,-0.2) rectangle (2.8,1.7);
  \node[above left] at (3.5,0) {$R$};
  \tikzstyle{every node}=[circle,draw,fill,minimum size=3pt,inner sep=0pt]
  \foreach\x in {0,0.6,1,1.4,1.7,2,2.3,2.8,3.5} \node at (\x,1) {};
\end{tikzpicture}
    \caption{An illustration of the initial set $\famP_0(R)$ of witness points.}
    %The fourth point from the left corresponds to a clique of size $3$ in $\famV(R)$.}
    \label{fig:witness}
\end{figure}

Now, let $i\in\{1,\ldots,k\}$, and let $u\in B_{i-1}$.
We partition the set $\famS_{i-1}(u)$ into two subsets $\famS_i(u0)$ and $\famS_i(u1)$, and we define the witness sets $\famP_i(R)$ for the rectangles $R\in\famS_i(u0)\cup\famS_i(u1)$, as follows.
We consider the rectangles $R\in\famS_{i-1}(u)$ in the order decreasing by height, so that all rectangles in $\famV(R)\cap\famS_{i-1}(u)$ are considered before $R$.
For each rectangle $R\in\famS_{i-1}(u)$ (in that order), if some witness point $p\in\famP_{i-1}(R)$ belongs to at least $2^{k-i}$ rectangles from $\famV(R)$ that have been already added to $\famS_i(w0)$, then we add $R$ to $\famS_i(w1)$ and let $\famP_i(R)$ be the set of all such witness points $p$; otherwise, we add $R$ to $\famS_i(w0)$ and let $\famP_i(R)=\famP_{i-1}(R)$.

\begin{proposition}
\label{prop:basic decomp}
We remark the following basic properties of this (hierarchical) decomposition.
\begin{itemize}
\item For each\/ $i=0,\ldots, k$, the family $\{\famS_i(w)\colon w\in B_i\}$ forms a partition of\/ $\famR$.
\item For each\/ $i=0,\ldots,k-1$ and each\/ $w\in B_i$, we have\/ $\famS_i(w)=\famS_{i+1}(w0)\cup\famS_{i+1}(w1)$.
\item For each rectangle\/ $R\in\famR$, we have\/ $\famP_0(R)\supseteq\cdots\supseteq\famP_k(R)$.
Moreover, if\/ $R\in\famS_{i+1}(w0)$, then\/ $\famP_i(R)=\famP_{i+1}(R)$.
\end{itemize}
\end{proposition}

Next, we prove some crucial properties of this decomposition.
We will use them in the design and analysis of our coloring algorithms.

\begin{lemma}
\label{lem:partition}
For each\/ $i=0,\ldots,k$, each\/ $w\in B_i$, each rectangle\/ $R\in\famS_i(w)$, every witness point in\/ $\famP_i(R)$ belongs to fewer than\/ $2^{k-i}$ rectangles in\/ $\famV(R)\cap\famS_i(w)$.
\end{lemma}

\begin{proof}
The proof goes by induction on $i$.
For the base case when $i=0$, every point in $R$ belongs to fewer than $2^k$ rectangles in $\famV(R)$, because $\famR$ has clique number at most $2^k$.
For the induction step, let $i\in\{1,\ldots,k\}$ and $w=u0$ or $w=u1$, where $u\in B_{i-1}$.
For the sake of contradiction, suppose a witness point $p\in\famP_i(R)$ belongs to at least $2^{k-i}$ rectangles in $\famV(R)\cap\famS_i(w)$.
If $w=u0$, then $R$ should be taken to $\famS_i(u1)$ instead of $\famS_i(u0)$, because $p$ is a witness point in $\famP_{i-1}(R)$ that belongs to at least $2^{k-i}$ rectangles in $\famS_i(u0)$ with height greater than the height of $R$.
If $w=u1$, then by the fact that $p$ is a witness point in $\famP_i(R)$, the point $p$ additionally belongs to at least $2^{k-i}$ rectangles in $\famV(R)\cap \famS_i(u0)$, so it belongs to at least $2^{k-i+1}$ rectangles in $\famV(R)\cap\famS_{i-1}(u)$, contradicting the induction hypothesis.
\end{proof}

\begin{lemma}
\label{lem:nested structure}
For each\/ $i=0,\ldots,k$, each\/ $w\in B_i$, any two rectangles\/ $R,R'\in\famS_i(w)$ such that\/ $R'\in\famX(R)$, and any witness point\/ $p'\in\famP_i(R')$, there is a witness point\/ $p\in\famP_i(R)$ that belongs to\/ $R'$ and to all rectangles in\/ $\famV(R')$ containing\/ $p'$.
\end{lemma}

\begin{proof}
The proof goes by induction on $i$.
For the base case when $i=0$, the claim follows from the fact that $\famV(R')\subseteq\famV(R)$ when $R'\in\famX(R)$: if $p'$ is the intersection point of the top side of $R'$ with the left or right side of some rectangle $R''\in\famV(R')\cup\{R'\}$, then we let $p$ be the intersection point of the top side of $R$ with the same side of $R''$.
See Figure~\ref{fig:aligned witness} for an illustration.

For the induction step, let $i\in\{1,\ldots,k\}$ and $w=u0$ or $w=u1$, where $u\in B_{i-1}$.
Let $R,R'\in\famS_i(w)$ be such that $R'\in\famX(R)$.
If $w=u0$, then the claim follows directly from the induction hypothesis, as $\famP_i(R)=\famP_{i-1}(R)$ and $\famP_i(R')=\famP_{i-1}(R')$.
Now, suppose $w=u1$.
Let $p'$ be a witness point in $\famP_i(R')$.
Thus $p'$ belongs to at least $2^{k-i}$ rectangles in $\famV(R')\cap\famS_i(u0)$.
By the induction hypothesis, since $p'\in\famP_{i-1}(R')$, there is a witness point $p\in\famP_{i-1}(R)$ that belongs to all rectangles in $\famV(R')$ containing $p'$.
In particular, $p$ belongs to at least $2^{k-i}$ rectangles in $\famV(R')\cap\famS_i(u0)$ and thus in $\famV(R)\cap\famS_i(u0)$, as $\famV(R')\subseteq\famV(R)$.
This shows that $p\in\famP_i(R)$.
\end{proof}

\begin{corollary}
\label{cor:witness point}
For each\/ $i=0,\ldots,k$, each\/ $w\in B_i$, and any two rectangles\/ $R,R'\in\famS_i(w)$ such that\/ $R'\in\famX(R)$, there is a witness point\/ $p\in\famP_i(R)$ that belongs to\/ $R'$.
\end{corollary}

\begin{proof}
Immediate from Lemma~\ref{lem:nested structure}, as the witness set $\famP_i(R')$ is always non-empty.
\end{proof}

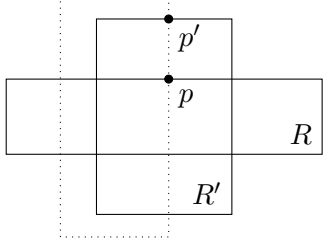
\begin{figure}[t]
    \centering
    \begin{tikzpicture}[xscale=1.2]
  \draw (0,0) rectangle (3.5,1);
  \draw (1,-0.8) rectangle (2.5,1.8);
  \draw[dotted] (0.6,-1.1) rectangle (1.8,2.1);
  \node[above left] at (3.5,0) {$R$};
  \node[above left] at (2.5,-0.8) {$R'$};
  \node[below right] at (1.8,1) {$p$};
  \node[below right,yshift=2pt] at (1.8,1.8) {$p'$};
  \tikzstyle{every node}=[circle,draw,fill,minimum size=3pt,inner sep=0pt]
  \foreach\y in {1,1.8} \node at (1.8,\y) {};
\end{tikzpicture}
    \caption{An illustration for the proof of Lemma~\ref{lem:nested structure}.
    In this figure, whenever $p'$ is a witness in $\famP_i(R')$, then $p$ is a witness in $\famP_i(R)$.}
    \label{fig:aligned witness}
\end{figure}

\begin{lemma}
\label{lem:clique}
For each\/ $i=0,\ldots,k$ and each rectangle\/ $R\in\famS_i(w)$, every clique in\/ $\famX(R)\cap\famS_i(w)$ has size at most\/ $2^{k-i+1}$.
\end{lemma}

\begin{proof}
For the sake of contradiction, suppose that there is a clique $\famC$ of size greater than $2^{k-i+1}$ in $\famX(R)\cap\famS_i(w)$.
Let $q$ be a point on the top side of $R$ that lies in all rectangles in $\famC$.
Let $\famC_L$ be the $2^{k-i}$ rectangles from $\famC$ with leftmost left sides.
Let $\famC_R$ be the $2^{k-i}$ rectangles from $\famC$ with rightmost right sides.
Let $R'$ be a rectangle in $\famC\setminus(\famC_L\cup\famC_R)$, which exists as $|\famC|>|\famC_L|+|\famC_R|$.
By Corollary~\ref{cor:witness point} applied to $R$ and $R'$, there is a witness point $p\in\famP_i(R)$ such that $p\in R'$.
If $p$ is to the left of $q$, then $p$ belongs to all rectangles in $\famC_L$ (as they contain $q$ and their left sides are more to the left than the left side of $R'$), which contradicts Lemma~\ref{lem:partition}.
An analogous contradiction is reached for $\famC_R$ if $p$ is to the right of $q$.
\end{proof}

%The idea behind the clique reduction originates from [] and is used also in [].

\subsection*{Step 2: Clique Reduction}

For $\alpha\in\setN$, an \emph{$\alpha$-covering} of a rectangle $R$ is a clique that contains $R$, at least $\alpha$ rectangles intersecting the top side of $R$, and at least $\alpha$ rectangles intersecting the bottom side of $R$ (not necessarily different from those intersecting the top side).

\begin{lemma}
\label{lem:alpha-covering}
Every clique in\/ $\famR$ of size greater than\/ $2\alpha$ is an\/ $\alpha$-covering of one of its rectangles.
\end{lemma}

\begin{proof}
Let $\famC$ be a clique in $\famR$ of size greater than $2\alpha$.
Let $\famC_T$ be the $\alpha$ rectangles in $\famC$ with top-most top sides.
Let $\famC_B$ be the $\alpha$ rectangles in $\famC$ with bottom-most bottom sides.
Let $R$ be a rectangle in $\famC\setminus(\famC_T\cup\famC_B)$, which exists as $|\famC|>|\famC_T|+|\famC_B|$.
The rectangles in $\famC_T$ and $\famC_B$ witness that $\famC$ is an $\alpha$-covering of $R$.
\end{proof}

For each $i=0,\ldots,k$ and each $w\in B_i$, let $\famT_i(w)$ be the set of rectangles $R\in\famS_i(w)$ such that there is a $2^{k-i+2}$-covering of $R$ in $\famS_i(w)$.
Observe that $\famT_i(w)=\emptyset$ when $i\leq 2$ and that it is easy to compute the sets $\famT_i(w)$ in polynomial time.
By Lemma~\ref{lem:alpha-covering}, at least one rectangle from every clique in $\famS_i(w)$ of size greater than $2^{k-i+3}$ belongs to $\famT_i(w)$, so $\omega(\famS_i(w)\setminus\famT_i(w))\leq 2^{k-i+3}$.

\begin{lemma}
\label{lem:sparse-covering}
For each\/ $i=0,\ldots,k$ and each\/ $w\in B_i$, the set\/ $\famT_i(w)$ is\/ $3$-sparse.
\end{lemma}

\begin{proof}
Let $R$ be a rectangle in $\famT_i(w)$.
We define three points $p_1^R,p_2^R,p_3^R\in R$ as follows.
Let $p_1^R$ be the leftmost point and $p_2^R$ be the rightmost point in the witness set $\famP_i(R)$.
Recall that these points lie on the top side of $R$.
Choose a clique $\famC$ that forms a $2^{k-i+2}$-covering of $R$ in $\famS_i(w)$, and let $p_3^R$ be a point in the intersection of all rectangles in $\famC$ (which include $R$).

We verify that these points capture all intersecting pairs of rectangles, as in the definition of sparseness.
Consider two crossing rectangles $R,R'\in\famT_i(w)$ such that $R'\in\famX(R)$.
We claim that at least one point of $p_1^R,p_2^R, p_3^R,p_1^{R'},p_2^{R'},p_3^{R'}$ lies in $R'\cap R$.
Suppose, for the sake of contradiction, that this is not the case.

First, observe that $R'$ cannot lie to the left of $p_1^R$: due to Corollary~\ref{cor:witness point}, this would mean that some witness point in $\famP_i(R)$ lies to the left of $p_1^R$, contradicting to the choice of $p_1^R$.
For the same reason, $R'$ cannot lie to the right of $p_2^R$, so it must lie between $p_1^R$ and $p_2^R$.

Since $p_3^{R'}\notin R$, the point $p_3^{R'}$ must lie either above or below $R$.
Assume that it lies below $R$ (see Figure~\ref{fig:sparse}; the other case is analogous, by symmetry).
Since $R'\in\famT_i(w)$, there is a $2^{k-i+2}$-covering $\famC$ of $R'$ in $\famS_i(w)$.
Let $\famC'$ be the rectangles in $\famC$ that intersect the top side of $R'$.
Thus $|\famC'|\geq 2^{k-i+2}$, by the definition of $2^{k-i+2}$-covering.
Since $R$ lies above $p_3^{R'}$ and below every point on the top side of $R'$, we have $\famC'\subseteq\famV(R)$.
Lemma~\ref{lem:clique} yields $|\famC'\cap\famX(R)| \leq 2^{k-i+1}$, so $|\famC'\setminus\famX(R)|\geq 2^{k-i+1}$.
Each rectangle in $\famC'\setminus\famX(R)$ fully contains the left or the right side of $R$ (or both).
Let $\famC'_1$ be the rectangles in $\famC'\setminus\famX(R)$ containing the left side of $R$, and let $\famC'_2$ be those containing the right side of $R$, so that $\famC'_1\cup\famC'_2=\famC'\setminus\famX(R)$.
It follows that $|\famC'_j|\geq 2^{k-i}$ for some $j\in\{1,2\}$.
This contradicts Lemma~\ref{lem:partition}, because $\famC'_j$ is a clique in $\famV(R)\cap\famS_i(w)$ of size at least $2^{k-i}$ containing the witness point $p_j^R\in\famP_i(R)$.
\end{proof}

\begin{figure}
    \centering
    \begin{tikzpicture}[xscale=1.2]
  \draw (0,0) rectangle (5,1);
  \draw (2,-1.2) rectangle (3.6,1.8);
  \draw[dashed] (-0.5,-1) rectangle (3,2.2);
  \draw[dashed] (1.7,-0.8) rectangle (5.6,2);
  \draw[dotted] (1,-0.6) rectangle (2.6,1.6);
  \node[above left] at (5,0) {$R$};
  \node[below left] at (3.6,1.8) {$R'$};
  \node[below] at (0.5,1) {$p_1^R$};
  \node[below] at (4.1,1) {$p_2^R$};
  \node[right] at (2.3,-0.3) {$p_3^{R'}$};
  \tikzstyle{every node}=[circle,draw,fill,minimum size=3pt,inner sep=0pt]
  \foreach\x in {0.5,4.1} \node at (\x,1) {};
  \node at (2.3,-0.3) {};
\end{tikzpicture}
    \caption{An illustration for the proof of Lemma~\ref{lem:sparse-covering}.
    The two dashed rectangles belong to $\famC'_1$ or $\famC'_2$; each such rectangle contains $p_1^R$ or $p_2^R$, respectively.
    The dotted rectangle belongs to $\famX(R)$.}
    \label{fig:sparse}
\end{figure}
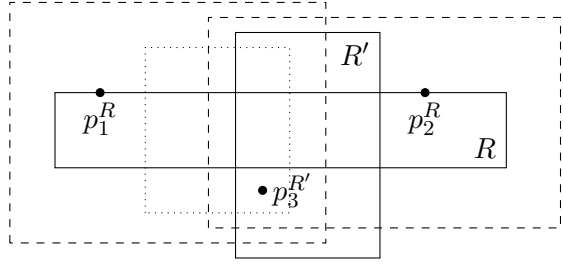

Finally, we present the coloring algorithm.
It proceeds in $k$ rounds.
Initially, we have $\famR_0=\famR$.
At round $i=1,\ldots, k$, we obtain $\famR_i$ by removing $\bigcup_{w\in B_i}\famT_i(w)$ from $\famR_{i-1}$.
After the last round, we are left with the set $\famR_k$.
Let $\famT'_i(w) =\famT_i(w) \cap \famR_{i-1}$ and $\famR'_i= \bigcup_{w\in B_i} \famT'_i(w)$ be the set of rectangles removed in round $i$.
Hence, the families $\famR'_1,\ldots,\famR'_k$ and $\famR_k$ form a partition of $\famR$.
We argue that each family in this partition can be colored using $O(2^k)$ colors.

For each $i=1,\ldots,k$ and each $w\in B_i$, by Lemma~\ref{lem:sparse-covering}, the family $\famT'_i(w)$ is $3$-sparse, and since $\famT'_i(w)\subseteq\famS_i(w)\cap\famR_{i-1}\subseteq\famS_{i-1}(u)\setminus\famT_{i-1}(u)$ (where $u$ is the prefix of $w$ in $B_{i-1}$) and $\omega(\famS_{i-1}(u)\setminus\famT_{i-1}(u))\leq 2^{k-i+4}$ (by the remark before Lemma~\ref{lem:sparse-covering}), the maximum size of a clique in $\famT'_i(w)$ is at most $2^{k-i+4}$.
Therefore, by Lemma~\ref{lem:sparse}, each set $\famT'_i(w)$ is $O(2^{k-i})$-colorable.
Since $|B_i|=2^i$, by coloring each set $\famT'_i(w)$ where $w\in B_i$ using a separate bunch of $O(2^{k-i})$ colors, we color $\famR'_i$ using $O(2^k)$ colors.
Finally, for each $w\in B_k$, we have $\famS_k(w)\cap\famR_k\subseteq\famS_k(w)\setminus\famT_k(w)$, so (by the remark before Lemma~\ref{lem:sparse-covering}) $\omega(\famS_k(w)\cap\famR_k)\leq\omega(\famS_k(w)\setminus\famT_k(w))\leq 8$, and therefore (by Lemma~\ref{lem:omega square}) the set $\famS_k(w)\cap R_k$ can be colored using $O(1)$ colors.
Again, using a separate bunch of $O(2^{k-i})$ colors on each set $\famS_k(w)\cap\famR_k$ where $w\in B_k$, we color the family $\famR_k$ using $O(2^k)$ colors.

\section{An \texorpdfstring{$O(\log\log n)$}{O(log log n)}-Approximation Algorithm for MWISR}
\label{sec:ind set}

In this section, we present a reduction from MWISR to the coloring problem, which leads to a polynomial-time $O(\log\log n)$-approximation algorithm for MWISR\@.
Similar reductions from the maximum weight independent set problem to the coloring problem have been already used in the literature of approximation algorithms \cite{chalermsook2011coloring,chalermsook2009maximum,lewin2002routing}.
We also show that the reduction can be made deterministic via a derandomization trick similar to the one used by Chan and Har-Peled~\cite{chan2012approximation}.

Let $\famR$ be a family of $n$ rectangles, and for each $R\in\famR$, let $w_R$ be the weight associated with $R$.
Assume that $w_R>0$ for each $R\in\famR$ (otherwise $R$ can be disregarded).
Let $\MaxCliques$ be the family of inclusion-maximal cliques in $\famR$.
Thus $|\MaxCliques|\leq n^2$, because the intersection of every such clique is a rectangle whose top left corner is the intersection point of the top side of some rectangle in $\famR$ with the left side of some (possibly the same) rectangle in $\famR$.
Consider the following clique-constrained LP relaxation of the maximum weight independent set problem:
\begin{alignat*}{2}
\text{maximize}&\quad\sum_{\!\!R\in\famR\!\!}w_Rx_R\mskip-50mu\\
\text{subject to}&\quad\sum_{\!R\in\famC\!}x_R&&\leq 1\quad\text{for every }\famC\in\MaxCliques,\\
&&\llap{$x_R$}&\geq 0\quad\text{for every }R\in\famR.
\end{alignat*}
Let $(x^\star_R)_{R\in\famR}$ be an optimal fractional solution to the LP, and let $w^\star$ be the optimum value, which is therefore an upper bound on the maximum weight of an independent set in $\famR$.
Let $m=\lceil 9\ln n\rceil$.

\begin{claim}
There is an integral vector\/ $(y_R)_{R\in\famR}\in\{0,\ldots,m\}^\famR$ such that
\[\sum_{R\in\famC}y_R\leq 2m\quad\text{for every}\enspace\famC\in\MaxCliques\quad\text{and}\quad\sum_{R\in\famR}w_Ry_R\geq\frac{mw^\star}{2}.\]
Moreover, such a vector can be computed by a polynomial-time deterministic algorithm.
\end{claim}

\begin{proof}
If $w_R\geq w^\star/2$ for some $R\in\famR$, then it is enough to set $y_R=m$ and $y_{R'}=0$ for $R'\in\famR\setminus\{R\}$.
Therefore, assume henceforth that $w_R<w^\star/2$ for every $R\in\famR$.

For each $R\in\famR$, let $x'_R=\lfloor mx^\star_R\rfloor$, let $x''_R$ be a random variable in $\{0,1\}$ such that $Ex''_R=P(x''_R=1)=mx^\star_R-x'_R$, and let $y_R=x'_R+x''_R$.
It follows that $0\leq x'_R\leq y_R\leq\lceil mx^\star_R\rceil\leq m$ for $R\in\famR$.
The LP inequality for a clique $\famC\in\MaxCliques$ yields
\[m\geq\sum_{R\in\famC}mx^\star_R=\sum_{R\in\famC}x'_R+E\biggl(\sum_{R\in\famC}x''_R\biggr),\]
which implies
\begin{align*}
P\biggl(\sum_{R\in\famC}y_R>2m\biggr)&\leq P\biggl(\sum_{R\in\famC}x'_R+\sum_{R\in\famC}x''_R>2\sum_{R\in\famC}x'_R+2E\biggl(\sum_{R\in\famC}x''_R\biggr)\biggr)\\
&\leq P\biggl(\sum_{R\in\famC}x''_R>2E\biggl(\sum_{R\in\famC}x''_R\biggr)\biggr)\\
&<\exp\biggl(-\frac{1}{3}E\biggl(\sum_{R\in\famC}x''_R\biggr)\biggr)\leq\exp\biggl(-\frac{m}{3}\biggr)\leq n^{-3},
\end{align*}
where the strict inequality is the following form of the Chernoff bound for a sum of independent zero-one random variables $z$: $P(z>2Ez)<\exp(-\frac{1}{3}Ez)$.
For each clique $\famC\in\MaxCliques$, let a random variable $\zeta_\famC$ be defined as follows:
\[\zeta_\famC=\begin{cases}
1&\text{if }\sum_{R\in\famC}y_R>2m,\\
0&\text{otherwise,}
\end{cases}\]
so that $E\zeta_\famC<n^{-3}$.
Let a random variable $\xi$ be defined as follows:
\[\xi=\sum_{R\in\famR}w_Ry_R-\frac{mnw^\star}{2}\sum_{\famC\in\MaxCliques}\zeta_\famC.\]
It follows that
\[E\xi=\sum_{R\in\famR}w_REy_R-\frac{mnw^\star}{2}\sum_{\famC\in\MaxCliques}E\zeta_\famC>m\sum_{R\in\famR}w_Rx^\star_R-\frac{mnw^\star}{2n^3}{|\MaxCliques|}\geq mw^\star-\frac{mw^\star}{2}=\frac{mw^\star}{2}.\]
Therefore, there is a choice of $(y_R)_{R\in\famR}$ where $\xi>mw^\star/2$.
It can be computed by a polynomial-time deterministic algorithm using the conditional expectation method, because  $E\zeta_\famC$ can be computed in polynomial time for every clique $\famC\in\MaxCliques$ by dynamic programming.
Moreover, whenever $\sum_{R\in\famC}y_R>2m$ for some $\famC\in\MaxCliques$, then $\xi<0$ (because $w_R<w^\star/2$ and $y_R\leq m$ for every $R\in\famR$), so the resulting choice of $(y_R)_{R\in\famR}$ satisfies the conditions of the claim.
\end{proof}

Now, let $(y_R)_{R\in\famR}$ be as in the claim.
Let $\famR'$ be a multiset of rectangles where each rectangle $R$ from $\famR$ occurs in $y_R$ copies.
The first condition of the claim implies that $\famR'$ has clique number at most $2m$, so it has chromatic number $O(m\log m)$, and moreover, a proper $O(m\log m)$-coloring of $\famR'$ can be computed in polynomial time.
The second condition of the claim implies that the rectangles in $\famR'$ have total weight at least $mw^\star/2$, so some color class in the aforesaid coloring of $\famR'$ has total weight $\varOmega(w^\star/\log m)=\varOmega(w^\star/\log\log n)$.
That color class can be returned as a requested $O(\log\log n)$-approximation of the maximum weight independent set in $\famR$.

\bibliographystyle{plain}
\bibliography{ref}

\end{document}